%% file: main.tex
\documentclass[10pt,twocolumn,twoside]{IEEEtran}
\usepackage[english]{babel}
\usepackage{amsmath,amssymb,amsfonts}
\usepackage{amsthm}
\usepackage{algorithm}
\usepackage{algorithmic}
\usepackage{textcomp}
\usepackage{xcolor}
\usepackage{tikz}
\usepackage{pgfplots}
\pgfplotsset{compat=newest}
\usepackage{subfigure}
\usepackage{float}
\usepackage{adjustbox}
\usepackage{svg}
\usepackage{cite}
\usepackage{enumitem}
\usetikzlibrary{positioning,chains,fit,shapes,calc}
\usetikzlibrary{backgrounds}
\usetikzlibrary{arrows.meta}
\usetikzlibrary{shapes}
\usetikzlibrary{decorations.pathmorphing}

\newtheorem{theorem}{Theorem}

\newtheorem{remark}{Remark}
\newtheorem{proposition}{Proposition}
\newtheorem{lemma}{Lemma}

\definecolor{ao}{rgb}{0.55, 0.71, 0.0}
\definecolor{bleudefrance}{rgb}{0.19, 0.55, 0.91}
\definecolor{dimgray}{rgb}{0.41, 0.41, 0.41}

\begin{document}

% Title
% \title{Steering Fractional Exponents in Fractional-Order Networks and its Implication in Neuroscience}
\title{Steering Fractional-Order Network Dynamics via Joint Parameter and State Control}

% Authors - IEEE format
\author{
\IEEEauthorblockN{Alessandro Varalda\IEEEauthorrefmark{1}}
\IEEEauthorblockA{\IEEEauthorrefmark{1}Division of Systems and Control, Uppsala University, alessandro.varalda@it.uu.se}\\
\and
\IEEEauthorblockN{Sérgio Pequito\IEEEauthorrefmark{2}}
\IEEEauthorblockA{\IEEEauthorrefmark{2}Institute for Systems and Robotics, Instituto Superior Técnico, University of Lisboa, sergio.pequito@tecnico.ulisboa.pt}
}

% Make the title area
\maketitle

\begin{abstract}
\input{0_Abstract}
\end{abstract}

% Redefining Controllability for Fractional-Order Networks and its Applications in Neuroscience
\begin{IEEEkeywords}
Fractional-order dynamical network, control theory, neuroscience, network systems.
\end{IEEEkeywords}

\section{Introduction}\label{sec:introduction}
\input{1_Introduction}

\section{Problem Statement}\label{sec:problem_statement}
\input{2_Problem_Statement}

\section{Control of the Discrete-time Fractional-Order Networks}\label{sec:main_results}
\input{3_Main_Results}

\section{Illustrative Examples}\label{sec:numerical_simulation}
\input{4_Numerical_Simulation}

\section{Conclusion and Future Research}\label{sec:conclusion}
\input{5_Conclusion}

% Bibliography
\bibliographystyle{IEEEtran}
\bibliography{references}

\end{document}

%% file: 0_Abstract.tex
%
%
%       Abstract
%
%
This paper studies the control of discrete-time linear fractional-order networks, a flexible modeling framework for systems with long-range memory such as power grids, biological networks, and neuronal circuits. In contrast to the common view that fractional exponents (time-scales) are fixed parameters, we show that they can be systematically steered, together with the network coupling matrix, by appropriately designed input sequences. 

We first derive algebraic conditions under which the coupling matrix and the vector of fractional exponents of a given network can be reconfigured to desired values, and we characterize how truncating the infinite-memory term impacts the resulting dynamics. Building on these results, we construct an equivalent linear representation that isolates the contribution of memory, and we introduce a fractional reachability matrix that provides explicit conditions for jointly steering both network parameters and state in a finite number of steps. 

To address practical implementations, we further formulate an energy-constrained steering problem that incorporates actuator bounds and finite-memory approximations as a quadratic program. The framework is illustrated on low-dimensional toy examples, on larger networks with Erd\H{o}s--R\'enyi, Barab\'asi--Albert, and Watts--Strogatz topologies, and on a brain network model inferred from electrocorticography recordings of an epilepsy patient, where we showcase transitions between pre-seizure and seizure configurations.

%% file: 1_Introduction.tex
%
%
%       Introduction
%
%
Fractional-order networks have gained widespread application in modeling, analyzing, and controlling various domains, including biological systems, electrical circuits, viscoelastic materials, and neuronal activity~\cite{Magin2004,Podlubny1999,Adolfsson2005,Lundstrom2008}. These networks are prized for their ability to express and explain phenomena through their unique feature: a fractional order that represents a power-law long-range memory effect on self-dynamics (i.e., a \emph{time-scale}). 

This characteristic enables these networks to model complex interactions which can not be captured by traditional models. Empirical studies demonstrate that a combination of network coupling parameters and time-scales can characterize different behaviors of the networks, for instance, in the context of brain networks, and in particular epilepsy, specific time-scales can serve as biomarkers of a abnormal behavior, associated with seizures~\cite{Reed2022,Reed2023,Ionescu2017}. 
% abnormal brain states, particularly during epileptic seizures, are characterized by specific fractional-order ranges that induce network instability. 

These observations suggest a promising approach to neurological interventions through precise control of both network coupling parameters and the different time-scales (i.e., fractional exponents).
%Developing methods to manipulate these parameters, therefore, represents a crucial bridge between theoretical research and potential clinical applications in neuroscience.

In this work, we take inspiration from classical control theory to define the conditions and methods that allow steering the parameters of a fractional-order network. Given that the network's future states are intrinsically a function of both its dynamic matrix and the current state, and since fractional exponents can change over time in response to specific inputs~\cite{Ortigueira2019}, we subsequently expand our method to provide steering conditions and methods for both state and parameters of a fractional-order network.

% Our paper contributes to this evolving field by (i) introducing a framework to steer fractional-network parameters, establishing the necessary conditions and input sequences for parameter manipulation; (ii) introducing a framework to steer fractional-network parameters and state simultaneously, providing a comprehensive approach for joint control; (iii) defining a controllability matrix for fractional-order systems, extending classical controllability concepts to accommodate fractional dynamics; and (iv) providing numerical and illustrative examples that demonstrate the effectiveness of our proposed frameworks.
Our main contributions are as follows:

\textbf{1. Input-based steering of fractional-order network parameters.}
    We derive conditions under which a discrete-time linear fractional-order network can be driven from an initial set of fractional time-scales and coupling strengths to a desired configuration that captures a target interaction pattern among the network nodes, by suitable input sequences. The approach separates and then combines the steering of the coupling structure via state feedback with the steering of the fractional time-scales via input modification, and provides explicit bounds on the error introduced by memory truncation.

 \textbf{2. Joint parameter--state steering via fractional reachability.}
    We construct an equivalent linear representation of the target fractional-order network that isolates the memory contribution as a structured disturbance. Using this representation, we introduce a fractional reachability matrix and obtain explicit conditions and closed-form expressions for input sequences that simultaneously steer network parameters and state to prescribed targets.

    \textbf{3. Energy-aware and constraint-aware control formulation.}
    We formulate an energy-minimization problem with actuator bounds and finite-memory approximations that embeds the above steering laws into a quadratic program, suitable for practical implementations in systems with limited control authority.

    \textbf{4. Validation on synthetic and neuroscience-inspired networks.}
    We validate the framework on synthetic networks with Erd\H{o}s-R\'enyi, Barab\'asi-Albert, and Watts-Strogatz topologies, demonstrating scalability beyond low-dimensional examples. We further apply our methodology to a brain network inferred from electrocorticography data of an epilepsy patient, illustrating transitions between pre-seizure and seizure-like configurations driven by parameter steering.

% The remain of this paper is organized as follows. Section II formalizes the problem statement, introducing two key challenges: steering fractional-order network parameters from initial to desired configurations (Problem 1) and simultaneously controlling both network parameters and network state (Problem 2). 

% Section III presents our main theoretical contributions, beginning with methods for steering coupling matrices through state feedback (Theorem 1), followed by techniques for controlling fractional exponents via input sequences (Theorem 2), and culminating in a unified framework for combined parameter steering (Theorem 3) and simultaneous parameter-state control (Theorem 4). 

% Section IV validates our theoretical framework through numerical examples, including a detailed illustration of simultaneous parameter and state steering in a 3-dimensional DTLFON, robustness analysis across different network topologies (Erdős-Rényi, Barabási-Albert, and Watts-Strogatz), and a practical demonstration using real electrocorticography data from epilepsy patients to showcase the transition between pre-seizure and seizure brain states. 

% Finally, Section V concludes with a discussion of the implications of our results and outlines promising directions for future research.

The remainder of this paper is organized as follows. Section~II introduces the problem formulation and two control tasks: steering fractional-order network parameters and jointly steering parameters and state. Section~III develops the theoretical framework enabling these tasks. Section~IV presents numerical studies on low-dimensional examples, synthetic network topologies, and a brain network model derived from electrocorticography data. Section~V concludes and outlines directions for future work.

%% file: 2_Problem_Statement.tex
%
%       Problem Statement
%
%
% A fundamental challenge in dynamical systems theory involves understanding how to transform one network model into another through appropriate interventions. This question becomes particularly relevant when dealing with systems whose underlying structure needs to be modified to achieve desired behaviors. 

% A recent class of models has gained special attention in the domain of modeling brain activity, due to converging evidence that complex neural dynamics can be effectively captured by the so-called (non-commensurate) discrete-time linear fractional-order networks (DTLFON)~\cite{Reed2022}. In this context, we are interested in understanding how transitioning from one DTLFON model to another is possible through strategic input sequences. 

Complex dynamical networks across diverse fields—from neural circuits and biological systems to communication networks and power grids—can be effectively modeled using discrete-time linear fractional-order networks (DTLFON), $\mathcal{N}_w(\alpha,A,B)$, described as follows:
\begin{align}  
    \begin{split}
        & \Delta^\alpha x[k+1] = Ax[k] + Bu[k],
        \label{eq:fosDynInput}
    \end{split}
\end{align} 
where $\Delta^\alpha x[k+1]=[\Delta^{\alpha_1}x_1[k+1] \ \ldots \  \Delta^{\alpha_n}x_n[k+1]]^\intercal$ is the fractional difference operator characterized by the fractional orders $\alpha=[\alpha_1 \ \ldots \ \alpha_n]^\intercal\in(\mathbb{R}_0^+)^n$ over the different entries of the state $x[k]\in\mathbb{R}^n$ and defined as 
\begin{equation}\label{eq:grunwald_letnikov}
    \Delta^{\alpha_i} x_i[k+1]= x_i[k+1] + \sum_{j=0}^{k} \psi(\alpha_i,j+1)x_i[k-j],
\end{equation} 
where $\psi(\alpha_i,j) = \frac{\Gamma(j - \alpha_i)}{\Gamma(-\alpha_i)\Gamma(j+1)}$, and $\Gamma(s) = \int_{0}^{\infty} t^{s-1}e^{-t} \,dt $, with $s\in\mathbb{R}$, is the Gamma function~\cite{Baleanu2011}. 

Although the Gr\"unwald–Letnikov operator is well defined for arbitrary real orders, in this work we restrict attention to nonnegative exponents $\alpha_i \ge 0$ to ensure that all Gamma functions appearing in the definition of $\psi(\alpha_i,j)$ are well-defined.
%since the Gamma function has poles at negative integers.

Additionally, the matrix $A\in \mathbb{R}^{n\times n}$ is the coupling matrix that describes the spatial relationship between the different entries of $x[k]$, with $k\in\mathbb{N}$ being the time step. Notably, $u[k]\in\mathbb{R}^p$ captures the $p$ inputs that change the dynamics, according to the matrix $B\in\mathbb{R}^{n\times p}$, at instance~$k$. 

\textbf{Problem} Let us consider the generic DTLFON at time instant $k \in\mathbb{N}$, $\mathcal{N}_w(\alpha[k],A[k],B)$, defined as
\begin{align}  
    \begin{split}
        \Delta^{\alpha[k]} x[k + 1] &= A[k]x[k] + Bu[k],
    \end{split}
    \label{eq:coupled_dynamics}  
\end{align}
where both $\alpha[k]\in (\mathbb{R}_0^+)^{n}$ and $A[k]\in \mathbb{R}^{n\times n}$ evolve over time in a way that is induced by the applied input sequence $\{u[j]\}_{j=0}^{k}$, so that different input sequences correspond to different trajectories of $(\alpha[k],A[k])$.

Given a fixed finite horizon $T\in\mathbb{N}$, an initial network $\mathcal{N}_w(\alpha[0],A[0],B)$ with initial state $x[0]=x_0$, and a desired target configuration $(\alpha_d,A_d,x_d)$, we aim to find an input sequence $\{u[j]\}_{j=0}^{T-1}$ which allows us to transition from 
\vspace{0.5pt}
\begin{align}  
    \begin{split}
        \mathcal{N}_w(\alpha[0],A[0],B) &\rightarrow \mathcal{N}_w(\alpha[T],A[T],B), \\
        x[0]=x_0 &\rightarrow x[T]=x_d,
    \end{split}
\end{align}
according to the dynamics in~\eqref{eq:coupled_dynamics}, where $\alpha[T] = \alpha_d$ and $A[T] = A_d$. Here, the symbol “$\rightarrow$” denotes the evolution induced by the input sequence $\{u[j]\}_{j=0}^{T-1}$ %under~\eqref{eq:coupled_dynamics}. 
In other words, the data $
\alpha[0],\; A[0],\; B,\; x_0,\; T,\; \alpha_d,\; A_d,\; x_d
$ 
are given, and the design variables are the inputs $\{u[j]\}_{j=0}^{T-1}$.

This problem addresses the simultaneous control of both model structure and network state. This challenge raises fundamental questions about the notion of controllability in fractional-order networks.

Notice that when all the fractional-order exponents in $\alpha_d$ are set to zero, the network becomes a linear time-invariant (LTI) network. In this case, our framework for DTLFON parameter controllability must align with the classic notion of controllability.

%% file: 3_Main_Results.tex
%
%
%       Main Results
%
%
In this section we develop the control-theoretic framework that enables the two tasks posed in Section~II: steering fractional-order network parameters and jointly steering parameters and state. In Section~\ref{subsec:steer_A} we show how to steer the coupling matrix $A$ of the DTLFON to a desired $A_d$ via state feedback (Theorem~\ref{thm:theorem1}). Section~\ref{subsec:steer_alpha} then addresses input-based steering of the fractional exponents $\alpha$, combining an order-perturbation analysis with conditions for reaching a desired $\alpha_d$ (Lemma~\ref{lem:lemma1}, Theorem~\ref{thm:theorem2}). Building on these ingredients, Section~\ref{subsec:combined} presents a unified control law that simultaneously steers $(A,\alpha)$ to $(A_d,\alpha_d)$ (Theorem~\ref{thm:theorem3}). Section~\ref{subsec:param_state} extends the framework to joint parameter–state steering by introducing an equivalent linear representation and a fractional reachability-based design that drives both $(A,\alpha)$ and $x$ to prescribed targets (Proposition~\ref{prop:defractionalization}, Theorem~\ref{thm:theorem4}). Finally, Section~\ref{subsec:energy_qp} incorporates actuator constraints and finite-memory implementation into an energy-aware quadratic-program formulation suitable for practical applications.

%First, we summarize DTLFON's key features to set up the important lemmas needed for our main results. We solved the presented problem by decomposing it into two sub-problems: (1) steering the coupling matrix $A$ using state feedback, outlined in Theorem~\ref{thm:theorem1}, and steering the fractional exponents $\alpha$, presented in Theorem~\ref{thm:theorem2}. We then combine these results to provide a comprehensive solution in Theorem~\ref{thm:theorem3}. 
%we addressed the problem of simultaneously achieving the desired network parameter configuration $(\alpha_d, A_d)$ while controlling the system state $x$ to reach a specific target value $x_d$. 
%(2) We addressed the comprehensive challenge of joint parameter-state steering, where network parameters $(A,\alpha)$ and network state $(x)$ are simultaneously guided to desired targets, highlighted in Theorem~\ref{thm:theorem4}.

\subsection{Steering Coupling Matrix ($A$)}\label{subsec:steer_A}
We begin by addressing the problem of steering the coupling matrix $A \in \mathbb{R}^{n \times n}$ to a desired configuration $A_d \in \mathbb{R}^{n \times n}$ through state feedback control.

\begin{theorem}\label{thm:theorem1}[Coupling Matrix Steering via State Feedback]
Given the network $\mathcal{N}_w(\alpha,A,B)$ with coupling matrix $A \in \mathbb{R}^{n \times n}$, the network can be steered to achieve a new configuration
\begin{align}  
    \begin{split}
        \mathcal{N}_w(\alpha,A_d,B):= \quad \Delta^{\alpha} x[k+1] = A_dx[k] + B\tilde{u}[k],
        \label{eq:network_with_desired_A}
    \end{split}
\end{align}
with desired coupling matrix $A_d \in \mathbb{R}^{n \times n}$ through state feedback if and only if there exists a feedback matrix $K \in \mathbb{R}^{p \times n}$ such that
\begin{align}
     BK = (A-A_d).
     \label{eq:coupling_steering}
\end{align}

The required feedback control law is given by
\begin{align}
    Bu[k] = B\tilde{u}[k] - BKx[k],
    \label{eq:feedback_control}
\end{align}
where $u[k] \in \mathbb{R}^p$ is the original input sequence in~\eqref{eq:fosDynInput} and $\tilde{u}[k] \in \mathbb{R}^p$ is input sequence of the desired network.
\hfill $\circ$
\end{theorem}

\begin{proof}
Consider the network $\mathcal{N}_w(\alpha,A,B)$ with the feedback control law in~\eqref{eq:feedback_control}:
\begin{align}
    \Delta^{\alpha}x[k+1] &= Ax[k] + Bu[k] \\
    &= Ax[k] + B(\tilde{u}[k] - Kx[k]) \nonumber \\
    &= Ax[k] + B\tilde{u}[k] - BKx[k] \nonumber \\
    &= (A - BK)x[k] + B\tilde{u}[k]
    \label{eq:feedback_dynamics}
\end{align}

For this network to have the desired dynamics $\Delta^{\alpha}x[k+1] = A_d x[k] + B\tilde{u}[k]$, it requires $A_d = A - BK$, from which it follows
\begin{align*}
    BK = (A-A_d).
\end{align*}

This equation is solvable for $K \in \mathbb{R}^{p \times n}$ if and only if each column of $(A_d - A)$ lies in the column space of $B$, that is, if $(A_d - A)e_i \in \mathrm{Im}(B)\subseteq \mathbb{R}^n$ for all $i=1,\dots,n$, where $\mathrm{Im}(B)$ denotes the image (column space) of $B$.\end{proof}

\begin{remark}\label{rmk:feedback_design}
The feedback matrix $K \in \mathbb{R}^{p \times n}$ can be designed using classical control techniques such as pole placement or optimal control methods, depending on the specific performance requirements and constraints of the application.
\hfill $\circ$
\end{remark}

\begin{remark}\label{rmk:input_constraints}
When $B \neq I_{n \times n}$, the achievable set of coupling matrices $A_d$ is constrained by the structure of the input matrix $B \in \mathbb{R}^{n \times p}$. This limitation reflects the physical constraints imposed by the actuator configuration.
\hfill $\circ$
\end{remark}

\subsection{Steering Fractional Exponents ($\alpha$)}\label{subsec:steer_alpha}

Towards understanding the impact of changing the fractional difference of order $\beta$ and considering~\eqref{eq:grunwald_letnikov}, we have the following lemma.

\begin{lemma}\label{lem:lemma1}[Decomposition of Fractional Differences under Order Perturbations]
Let $\beta,\delta\beta\in\mathbb{R}_0^+$, the fractional difference of order $\beta+\delta\beta\in(\mathbb{R}_0^+)$ applied to $z\in\mathbb{R}$ at time $k+1\in\mathbb{N}$ can be expressed as
\begin{align}
    \begin{split}
        \Delta^{\beta + \delta \beta}z[k+1] &= z[k+1]+\sum_{j=0}^{k} \psi_{\beta + \delta \beta,j+1}z[k-j], 
        \label{eq:grunwald_letnikov_expand}
    \end{split}
\end{align}
with  
\begin{align}
    \begin{split}
        \psi_{\beta + \delta \beta,j} = \left\{\begin{array}{ll}
             & 1, \quad  j=0 \\
             & \psi_{\beta,j}+\psi_{\delta\beta,j}+e_{\beta,\delta\beta,j}, \quad  j>0,
        \end{array}\right.
        \label{eq:psi_funct_decomp}
    \end{split}
\end{align}
where $\psi_{\beta,j}\in\mathbb{R}$ and $\psi_{\delta\beta,j}\in\mathbb{R}$ are the weights inherent to $\beta$ and $\delta\beta$, respectively. 

The term $e_{\beta,\delta\beta,j+1}\in\mathbb{R}$ is a residual encompassing the cross-products of $\beta$ and $\delta\beta$ described by
\begin{align*}
e_{\beta, \delta \beta,j+1} = -\psi_{\beta,j}\frac{\delta\beta}{j+1}&-\psi_{\delta\beta,j}\frac{\beta}{j+1}\\
&+e_{\beta, \delta \beta,j}\left( \frac{j-(\beta+\delta\beta)}{j+1} \right),
\end{align*}
with $e_{\beta,\delta\beta,1} = 0$.
% Moreover, we have
% \begin{equation*}
% |e_{\beta,\delta\beta,j}| \le c(\beta,\delta\beta)\, j^{-\min\{\beta,\delta\beta\}-1},\quad j\ge 1,
% \end{equation*}
% for an explicit constant $c(\beta,\delta\beta)$ obtained from the recursion and Stirling-type bounds for ratios of Gamma functions (see, e.g., \cite[Ch.~13]{Arfken2013}; a detailed derivation follows from applying these bounds to the recursion~\eqref{eqn:proof3}).

\hfill $\circ$
\end{lemma}

\begin{proof}    
By replacing $\beta$ by $\beta+\delta\beta$ in $\psi(\beta,j)$ from~\eqref{eq:grunwald_letnikov}, we have
\begin{align}
    \begin{split}
        \psi_{\beta + \delta \beta,j} = \frac{\Gamma( j - (\beta + \delta \beta) )}{\Gamma(- (\beta + \delta \beta) )\Gamma(j+1)},
        \label{eq:proof1}
    \end{split}
\end{align}
and by invoking the following property of the Gamma function from~\cite{Arfken2013}, 
$\Gamma(p+1)=p\Gamma(p)$,
it is possible to rewrite~\eqref{eq:proof1} as
$$\psi_{\beta + \delta \beta,j} = \frac{\Gamma( j-1- (\beta + \delta \beta) )(j-1- (\beta + \delta \beta))}{\Gamma(- (\beta + \delta \beta) )\Gamma(j) j}$$
$$=\psi_{\beta + \delta \beta,j-1}\frac{(j-1- (\beta + \delta \beta))}{j}$$
$$=\psi_{\beta + \delta \beta,j-1}\left(\frac{j-1-\beta}{j}+\frac{j-1-\delta\beta}{j}+\frac{1-j}{j}\right)$$
$$=\psi_{\beta + \delta \beta,j-1}\left(\frac{j-1-\beta}{j}\right)+\psi_{\beta + \delta \beta,j-1}\left(\frac{j-1-\delta\beta}{j}\right) $$
$$+\psi_{\beta + \delta \beta,j-1}\left(\frac{1-j}{j}\right).$$
% \begin{equation}\label{eq:proof2}\tag{*}\end{equation}

Moreover, one can repeat the same decomposition presented above for all $\psi_{\beta + \delta \beta,j}$ down to $j-1$.  Ultimately, $\psi_{\beta + \delta \beta,j}$ can be represented as the sum of the following three terms:
\begin{enumerate}
    \item[(1$^{\text{st}})$] $(\frac{j-1-\beta}{j})(\frac{j-2-\beta}{j-1})\ldots(\frac{-\beta}{1})=\psi_{\beta,j}$,
    \item[(2$^{\text{nd}})$] $(\frac{j-1-\delta\beta}{j})(\frac{j-2-\delta\beta}{j-1})\ldots(\frac{-\delta\beta}{1})=\psi_{\delta\beta,j}$, and
    \item[(3$^{\text{rd}})$] denoted as $e_{\beta,\delta\beta,j}$, which entails the remaining residual cross-product terms.
\end{enumerate}
Hence, \[ \psi_{\beta + \delta \beta,j} = \psi_{\beta,j}+\psi_{\delta\beta,j}+e_{\beta,\delta\beta,j}. \]

Furthermore, we can leverage this result in order to compute $e_{\beta,\delta\beta,j+1}$, in fact

\[
\psi_{\beta + \delta \beta,j+1}=\psi_{\beta + \delta \beta,j}\left(\frac{j-\beta}{j+1}+\frac{j-\delta\beta}{j+1}-\frac{j}{j+1}\right) 
\]
\[
= (\psi_{\beta,j}+\psi_{\delta \beta,j}+e_{\beta, \delta \beta,j})\left(\frac{j-\beta}{j+1}+\frac{j-\delta\beta}{j+1}-\frac{j}{j+1}\right)
\]
\[
=\psi_{\beta,j}\frac{j-\beta}{j+1} + \psi_{\delta \beta,j}\frac{j-\delta\beta}{j+1}+\psi_{\beta,j}\left(\frac{j-\delta\beta}{j+1}-\frac{j}{j+1}\right)
\]
{\footnotesize \[
+\psi_{\delta\beta,j}\left(\frac{j-\beta}{j+1}-\frac{j}{j+1}\right)+e_{\beta, \delta \beta,j}\left(\frac{j-\beta}{j+1}+\frac{j-\delta\beta}{j+1}-\frac{j}{j+1}\right)
\] }
\[
= \psi_{\beta,j+1}+\psi_{\delta\beta,j+1}-\psi_{\beta,j}\frac{\delta\beta}{j+1}-\psi_{\delta\beta,j}\frac{\beta}{j+1}
\]
\[
\qquad\qquad\qquad\qquad\qquad\qquad + e_{\beta, \delta \beta,j}\left( \frac{j-(\beta+\delta\beta)}{j+1} \right),
\]

\noindent from which it follows that
\begin{align}
    \begin{split}\label{eqn:proof3}
    e_{\beta, \delta \beta,j+1} &= -\psi_{\beta,j}\frac{\delta\beta}{j+1}-\psi_{\delta\beta,j}\frac{\beta}{j+1}+\\
    & \qquad\qquad e_{\beta, \delta \beta,j}\left( \frac{j-(\beta+\delta\beta)}{j+1} \right).
    \end{split}
\end{align}

Notice that an exceptional case arises when $j=0$, leading to $\psi_{\beta + \delta \beta,j}$ to be equal to
\[
 \psi_{\beta + \delta \beta,0} = \frac{\Gamma( 0 - (\beta + \delta \beta) )}{\Gamma(- (\beta + \delta \beta) )\Gamma(0+1)}=\frac{1}{\Gamma(1)}=1,
\]
while, for $j=1$, we have
\[ %\hspace{-3cm}
 \psi_{\beta + \delta \beta,1} = \frac{\Gamma( 1 - (\beta + \delta \beta) )}{\Gamma(- (\beta + \delta \beta) )\Gamma(1+1)}
\]
{\footnotesize\[
 = \frac{-(\beta + \delta \beta)\Gamma(-(\beta + \delta \beta))}{\Gamma(- (\beta + \delta \beta) )\Gamma(2)}=-\frac{\beta}{\Gamma(2)}-\frac{\delta\beta}{\Gamma(2)}=\psi_{\beta,1}+\psi_{\delta\beta,1},
\]}
from which we get that $e_{\beta,\delta\beta,1} = 0$.
\end{proof}

\begin{remark}\label{rmk:zero_cross_product}
Whenever $\beta$ or $\delta\beta$ in~\eqref{eq:psi_funct_decomp} are zero all \mbox{cross-product} terms disappear, therefore it follows that
\[
e_{0,\delta\beta,j}=e_{\beta,0,j}=0, \qquad \forall j\in\mathbb{N}.
\]
\hfill $\circ$
\end{remark}

% \begin{remark}\label{rmk:intro_Dmatrix}
% By applying Lemma~\ref{lem:lemma1} in a element-wise fashion, it follows that
% \begin{align}
%     \begin{split}
%         \Delta^{\beta + \delta \beta}z[k+1] &= z[k+1]+\sum_{j=0}^{k} \psi_{\beta + \delta \beta,j+1}z[k-j], 
%         \label{eq:grunwald_letnikov_expand_Dmatrix}
%     \end{split}
% \end{align}
% where the fractional difference of order $\beta+\delta\beta\in\mathbb{R}^n$ is applied to $z\in\mathbb{R}^n$ at time $k+1\in\mathbb{N}$.
% \end{remark}

Next, by leveraging Lemma~\ref{lem:lemma1} in the context of DTLFON, we provide the solution for steering fractional exponents.

\begin{theorem}\label{thm:theorem2}[Fractional Exponent Steering via Input Modification]
Given the network $\mathcal{N}_w(\alpha,A,B)$ with initial fractional exponents $\alpha\in(\mathbb{R}_0^+)^n$, we can obtain a desired network 
\begin{align}  
    \begin{split}
        \mathcal{N}_w(\alpha_d,A,B):= \quad \Delta^{\alpha_d} x[k+1] = Ax[k] + B\tilde{u}[k],
        \label{eq:network_with_desired_alpha}
    \end{split}
\end{align}
with desired fractional exponents $\alpha_d=\alpha+\delta\alpha\in(\mathbb{R}_0^+)^n$ if the input matrix $B \in \mathbb{R}^{n \times p}$ respects either of the following conditions
\begin{enumerate}[label=(\roman*)]
    \item $\textrm{rank}(B)=n$ (i.e., full row rank), or
    \item $\sum_{j=0}^k \tilde{D}^{\alpha,\delta\alpha}_{j+1} x[k-j] \in \text{Im}(B)$  for all $k=0,\dots,T-1$ if $\textrm{rank}(B)<n$.
\end{enumerate}
% \textcolor{blue}{As in the problem formulation, we restrict attention to nonnegative orders and consider $\alpha,\alpha_d\in(\mathbb{R}_0^+)^n$, so that the Gr\"unwald--Letnikov operator and the bounds of Lemma~\ref{lem:lemma1} are well defined componentwise.}
Notably, $\tilde{D}^{\alpha,\delta\alpha}_{j}\in \mathbb{R}^{n\times n}$ is a diagonal matrix which diagonal terms are function of  ($\alpha,\delta\alpha,j$) 
\[
\tilde{D}^{\alpha,\delta\alpha}_{j} =\text{diag}(\psi_{\delta\alpha_1,j}+e_{\alpha_1,\delta\alpha_1,j}, \ \ \ldots,\ \ \psi_{\delta\alpha_n,j}+e_{\alpha_n,\delta\alpha_n,j}),
\]
where the additional terms are defined as in Lemma~\ref{lem:lemma1}.

When one of these conditions is satisfied, the sequence of $T\in\mathbb{N}$ inputs $\{u[k]\}_{k=0}^{T-1}$ required to steer the fractional-exponent of a network, to a desired value, is given by
\begin{align}
\begin{split}
    Bu[k] = B\tilde{u}[k] - \sum_{j=0}^{k}\tilde{D}^{\alpha,\delta\alpha}_{j+1}x[k-j],
    \label{eq:alpha_steering_control}
\end{split}
\end{align}
where $u[k] \in \mathbb{R}^p$ is the input of the original network $\mathcal{N}_w(\alpha,A,B)$ and $\tilde{u}[k] \in \mathbb{R}^p$ is the input sequence of the desired network $\mathcal{N}_w(\alpha_d,A,B)$.
\hfill $\circ$
\end{theorem}

\begin{proof}
First, let us exploit $\Delta^{\alpha + \delta \alpha}x[k+1]$, by leveraging~\eqref{eq:grunwald_letnikov}   variable-wise, it readily follows that
\[
\Delta^{\alpha + \delta \alpha}x[k+1] = x[k+1] + \sum_{j=0}^{k} (\tilde{D}^{0,\alpha+\delta\alpha}_{j+1})x[k-j],
\]
where $
\tilde{D}^{0,\alpha + \delta\alpha}_{j} =\text{diag}(\psi_{\alpha_1+\delta\alpha_1,j},\ \ \ldots,\ \ \psi_{\alpha_n+\delta\alpha_n,j})$. Moreover, by applying Lemma~\ref{lem:lemma1} it follows that $
\tilde{D}^{0,\alpha + \delta\alpha}_{j} = \tilde{D}^{0,\alpha}_{j} + \tilde{D}^{\alpha,\delta\alpha}_{j}$ and consequentially
\[
\Delta^{\alpha + \delta \alpha}x[k+1] = x[k+1] + \sum_{j=0}^{k} (\tilde{D}^{0,\alpha}_{j+1} + \tilde{D}^{\alpha,\delta\alpha}_{j+1})x[k-j].
\]
Next, observe that
\begin{align}
\Delta^{\alpha + \delta \alpha}x[k+1] &= x[k+1] + \sum_{j=0}^{k}\tilde{D}^{0,\alpha}_{j+1}x[k-j] \nonumber\\
&\quad + \sum_{j=0}^{k}\tilde{D}^{\alpha,\delta\alpha}_{j+1}x[k-j], \nonumber
\end{align}

and adopting the expression presented in~\eqref{eq:grunwald_letnikov}, it follows that
\begin{equation}\label{eq:proof2lemma1}
    \Delta^{\alpha + \delta \alpha}x[k+1] = \Delta^{\alpha}x[k+1] + \sum_{j=0}^{k}\tilde{D}^{\alpha,\delta\alpha}_{j+1}x[k-j]. 
\end{equation}     

Now, by replacing $\Delta^{\alpha}x[k+1]$ with~\eqref{eq:fosDynInput}, it follows that
    \[
    \Delta^{\alpha + \delta \alpha}x[k+1]= Ax[k] + Bu[k]  + \sum_{j=0}^{k}\tilde{D}^{\alpha,\delta\alpha}_{j+1}x[k-j],
    \]
    from which we can take $B\tilde{u}[k]=Bu[k]  + \sum_{j=0}^{k}\tilde{D}^{\alpha,\delta\alpha}_{j+1}x[k-j]$ to obtain the desired network $\mathcal{N}_w(\alpha_d,A,B)$.
\end{proof}

% \begin{remark}\label{rmk:truncation}
%    The implementation of equation~\eqref{eq:alpha_steering_control} requires maintaining state history, characteristic of fractional-order networks. Following~\cite{Reed2022}, the infinite sum can be truncated to $J$ terms ($J >> 0$) since coefficients $\psi(\alpha_i, j)$ decay as $\mathcal{O}(j^{-\alpha_i-1})$, enabling accurate approximation without significant loss of precision as long as the state remains bounded.
   
% \hfill $\circ$
% \end{remark}

\begin{proposition}[Memory truncation error]
Let $\alpha_i, \delta\alpha_i \in \mathbb{R}^+_0$ be the i-th entry of $\alpha$ and $\delta\alpha \in (\mathbb{R}^+_0)^n$, and assume each state component to be bounded, i.e., $|x_i[k]| \le b_{x,i}$, with $b_{x,i} \in \mathbb{R}^+$. Truncating $\sum_{j=0}^{\infty}\tilde D^{\alpha,\delta\alpha}_{j+1}x[k-j]$ to $J\in \mathbb{N}$ terms yields tail
\begin{equation*}
\Big|\sum_{j=J+1}^{\infty}\tilde{d}_{\alpha_i,\delta\alpha_i,j+1}\,x_i[k-j]\Big|
\le b_{x,i}\,\tilde c\, \delta\alpha_i \sum_{j=J+1}^\infty j^{-\alpha_i+1} \ln j,
\end{equation*}
where $\tilde{d}_{\alpha_i,\delta\alpha_i,j} = \psi_{\delta\alpha_i,j} + e_{\alpha_i,\delta\alpha_i,j}$ is the $i$-th diagonal entry of $\tilde{D}^{\alpha,\delta\alpha}_{j}$, and $\tilde c \in \mathbb{R}$ is a fixed constant. In particular, if $\alpha_i\ge \underline{\alpha}>0$ for all $i$, then
\[
\sum_{j=J+1}^\infty j^{-\alpha_i+1}\ln j \le C\,J^{-\underline{\alpha}+1}\ln J
\]
for some constant $C>0$, so the truncation error decays at least on the order of $J^{-\underline{\alpha}+1}\ln J$ as $J\to\infty$.
\end{proposition}
% = O(\delta\alpha \cdot J^{-\alpha+1} \ln J)   %% for future purpose

\begin{proof}
By the triangle inequality and the bound $|x_i[k-j]| \le b_{x,i}$, we have
\begin{align*}
\Big|\sum_{j=J+1}^{\infty}\tilde{d}_{\alpha_i,\delta\alpha_i,j+1}\,x_i[k-j]\Big|
&\le \sum_{j=J+1}^{\infty}|\tilde{d}_{\alpha_i,\delta\alpha_i,j+1}| \, |x_i[k-j]|\\
&\le b_{x,i}\sum_{j=J+1}^{\infty}|\tilde{d}_{\alpha_i,\delta\alpha_i,j+1}|.
\end{align*}

From Lemma 1, we have 
\begin{equation*}
\tilde{D}^{\alpha,\delta\alpha}_{j} = \tilde{D}^{0,\alpha+\delta\alpha}_{j} - \tilde{D}^{0,\alpha}_{j},
\end{equation*}
with i-th component
\begin{equation*}
\tilde{d}_{\alpha_i,\delta\alpha_i,j} = \psi_{\delta\alpha_i,j} + e_{\alpha_i,\delta\alpha_i,j}.
\end{equation*}

From~\cite{Alessandretti2020,Sopasakis2015} we have that each of the i-th entries of $\tilde{D}^{0,\alpha}_j$ can be computed in the order $\mathcal{O}(j^{-\alpha_i+1})$, it follows that
\begin{align*}
\tilde{d}_{\alpha_i,\delta\alpha_i,j} &= \mathcal{O}(j^{-(\alpha_i+\delta\alpha_i)+1}) - \mathcal{O}(j^{-\alpha_i+1})\\
&= j^{-(\alpha_i+\delta\alpha_i)+1} - j^{-\alpha_i+1},
\end{align*}
by factoring out $j^{-\alpha_i+1}$, it follows
\begin{equation*}
\tilde{d}_{\alpha_i,\delta\alpha_i,j} = j^{-\alpha_i+1}(j^{-\delta\alpha_i} - 1).
\end{equation*}

Notably, for small $\delta\alpha_i$, using the Taylor expansion 
\begin{equation*}
j^{-\delta\alpha_i} = e^{-\delta\alpha_i \ln j} \approx 1 - \delta\alpha_i \ln j + \mathcal{O}((\delta\alpha_i)^2),
\end{equation*}
we obtain
\begin{equation*}
\tilde{d}_{\alpha_i,\delta\alpha_i,j} = j^{-\alpha_i+1}(-\delta\alpha_i \ln j + \mathcal{O}((\delta\alpha_i)^2)).
\end{equation*}

Therefore, for sufficiently large $j$ we obtain
\begin{equation*}
|\tilde{d}_{\alpha_i,\delta\alpha_i,j+1}| \le \tilde{c} \, \delta\alpha_i \, j^{-\alpha_i+1} \ln j,
\end{equation*}
for some constant $\tilde{c} \in \mathbb{R}$. 

Consequently:
\begin{equation*}
\Big|\sum_{j=J+1}^{\infty}\tilde{d}_{\alpha_i,\delta\alpha_i,j+1}\,x_i[k-j]\Big|
\le b_{x,i}\,\tilde{c}\, \delta\alpha_i \sum_{j=J+1}^\infty j^{-\alpha_i+1} \ln j.
\end{equation*}
\end{proof}

\begin{remark}\label{rmk:alpha_space}
   Notably, the input sequence in~\eqref{eq:alpha_steering_control} may lead to a change in either $x[k]\in\mathbb{R}^n$ and/or in the fractional exponents~$\alpha\in\mathbb{R}^n$. It is noteworthy to notice that, in general, not all the elements of $\alpha$ necessarily require to be steered. 
   % For instance, as demonstrated in \cite{Reed2023}, stabilizing a DTLFON may only entail significant alterations in a subset of $\alpha$.

   Moreover, the reachable space determined by~\eqref{eq:alpha_steering_control} is computed as a combination of the input matrix $B$ and $\alpha$. Hence, depending on how many $\alpha$ are required to be steered, it is possible for the matrix $B$ not to have full-rank.
\hfill $\circ$
\end{remark}

\subsection{Combined Steering of Coupling Matrix and Fractional Exponents ($\alpha,A$)}\label{subsec:combined}

We now present the main result that combines the steering of both coupling matrix and fractional exponents.

\begin{theorem}\label{thm:theorem3}[Combined Steering of Coupling Matrix and Fractional Exponents]
Given the DTLFON $\mathcal{N}_w(\alpha, A, B)$ in~\eqref{eq:fosDynInput}, the network can be steered to the desired configuration $\mathcal{N}_w(\alpha_d, A_d, B)$ through a combined control strategy if the following conditions are satisfied:
\begin{enumerate}[label=(\roman*)]
    \item Theorem~\ref{thm:theorem1} conditions hold for steering $A \rightarrow A_d$, and
    \item Theorem~\ref{thm:theorem2} conditions hold for steering $\alpha \rightarrow \alpha_d$.
\end{enumerate}
The combined control input sequence $\{u[k]\}_{k=0}^{T-1}$ is given by
\begin{align}
    Bu[k] = B\tilde{u}[k] - \sum_{j=0}^k \tilde{D}^{\alpha,\delta\alpha}_{j+1} x[k-j] - BKx[k],
    \label{eq:combined_control}
\end{align}
where $\tilde{u}[k] \in \mathbb{R}^p$ is the input of the desired network, the second term handles fractional exponent steering (from Theorem~\ref{thm:theorem2}), and the third term provides state feedback for coupling matrix steering (from Theorem~\ref{thm:theorem1}).
\hfill $\circ$
\end{theorem}

\begin{proof}
The proof follows directly from the combination of Theorems~\ref{thm:theorem1} and~\ref{thm:theorem2}. The coupling matrix steering requires the state feedback from Theorem~\ref{thm:theorem1}, while the fractional exponent steering requires the input modification from Theorem~\ref{thm:theorem2}. Since these modifications are additive in the input space, the combined control law~\eqref{eq:combined_control} achieves both objectives simultaneously, provided the individual conditions are satisfied.
\end{proof}

\begin{remark}\label{rmk:reachable_set}
The reachable set of $(\alpha_d, A_d)$ configurations is determined by the intersection of the individual reachable sets from Theorems~\ref{thm:theorem1} and~\ref{thm:theorem2}. This intersection depends on the structure of the input matrix $B \in \mathbb{R}^{n \times p}$ and the initial network parameters.
\hfill $\circ$
\end{remark}

\subsection{Network Parameter Steering with Simultaneous State Control ($\alpha,A,x$)}\label{subsec:param_state}

Having established methods for steering both the coupling matrix $A$ and fractional exponents $\alpha$, we now address the comprehensive problem of simultaneously achieving the desired network parameter configuration $(\alpha_d, A_d)$ while controlling the network state $x$ to reach a specific target value $x_d$. 

This approach recognizes that in practical applications, we often need to both reconfigure the network dynamics through parameter steering and guide the network state to a desired operating point. First, we establish a fundamental relationship that enables this combined approach.

\begin{proposition}\label{prop:defractionalization}[Equivalent Linear Representation of Fractional Networks]
Consider the DTLFON $\mathcal{N}_w(\alpha_d,A_d,B)$, it can be represented in an equivalent form that separates the linear dynamics from the fractional memory effects
\begin{equation}\label{eq:defractionalized}
    x[k+1]=A_d x[k]+B\tilde{u}[k]-\sum_{j=0}^{k} \tilde{D}^{0,\alpha_d}_{j+1}x[k-j],
\end{equation}
where $A_d \in \mathbb{R}^{n \times n}$ is the desired coupling matrix and the term $\sum_{j=0}^{k} \tilde{D}^{0,\alpha_d}_{j+1}x[k-j]$ captures the fractional memory effects.
\hfill $\circ$
\end{proposition}

\begin{proof}
Starting from the desired fractional network $\Delta^{\alpha_d}x[k+1] = A_d x[k] + B\tilde{u}[k]$, we apply the decomposition from Lemma~\ref{lem:lemma1} with $\beta = 0$ and $\delta\beta = \alpha_d$:
\[
\Delta^{0+\alpha_d} x[k+1]= \Delta^0x[k+1]+\sum_{j=0}^{k}\tilde{D}^{0,\alpha_d}_{j+1}x[k-j],
\]
which, when substituted into the desired network equation, yields:
\[
\Delta^0x[k+1]+\sum_{j=0}^{k}\tilde{D}^{0,\alpha_d}_{j+1}x[k-j]= A_d x[k]+B\tilde{u}[k],
\]
with $\Delta^0x[k+1]=x[k+1]$ from~\cite{Reed2022}, we obtain~\eqref{eq:defractionalized}.
\end{proof}

\begin{remark}\label{rmk:defractionalization_interpretation}
Proposition~\ref{prop:defractionalization} reveals that steering network parameters $(\alpha_d, A_d)$ while achieving the desired state $x_d$ is equivalent to controlling a linear network with coupling matrix $A_d$ subject to a time-varying disturbance term that depends on the fractional memory effects.
\hfill $\circ$
\end{remark}

Building on this proposition, we present the main result for simultaneous network parameter steering and state control.

\begin{theorem}\label{thm:theorem4}[Simultaneous Parameter and State Control]
Consider the DTLFON $\mathcal{N}_w(\alpha, A, B)$ in~\eqref{eq:fosDynInput} with initial state $x_0 = x[0]$. The network can be steered to the desired parameter configuration $\mathcal{N}_w(\alpha_d, A_d, B)$ while simultaneously driving the state from $x_0$ to a desired final state $x_d = x[T]$ in $T$ time steps if and only if there exists a sequence of inputs $\{\tilde{u}[k]\}_{k=0}^{T-1}$ satisfying the combined control law from Theorem~\ref{thm:theorem3} such that 

\vspace{0.5pt}
\begin{align}
    \begin{split}
        \{\tilde{u}[k]\}_{k=0}^{T-1} = {C^{0,\alpha_d}_{T}}^{\dag} \left( x_d- G^{0,\alpha_d}_{T}x_0\right),
    \label{eq:state_control_solution}
    \end{split}
\end{align} 
here the fractional reachability matrix $C^{0,\alpha_d}_{k}\in\mathbb{R}^{n\times kp}$, with ${C^{0,\alpha_d}_{T}}^{\dag} = ( (C^{0,\alpha_d}_{T})^\intercal C^{0,\alpha_d}_{T})^{-1}(C^{0,\alpha_d}_{T})^\intercal$, is given by
\[
  C^{0,\alpha_d}_{T} = \begin{bmatrix} G^{0,\alpha_d}_{0}B & G^{0,\alpha_d}_{1}B & \ldots & G^{0,\alpha_d}_{T-1}B \end{bmatrix}.
\]

The fractional state transition matrix $G^{0,\alpha_d}_{k}\in\mathbb{R}^{n\times n}$ is defined recursively as
\[
    G^{0,\alpha_d}_{k}=\left\{\begin{array}{ll}
             & \mathbb{I}_{n\times n}, \quad\quad\quad\quad\quad\quad\quad\quad  k=0, \\
             & \sum_{j=0}^{k-1} \tilde{A}^{0,\alpha_d}_{d,j} G^{0,\alpha_d}_{k-1-j}, \quad  k>0,
        \end{array}\right.
\]
where the modified network matrix $\tilde{A}^{0,\alpha_d}_{d,j}\in\mathbb{R}^{n \times n}$ incorporates both the desired coupling matrix and fractional effects as follows
\[
\tilde{A}^{0,\alpha_d}_{d,j}=\left\{\begin{array}{ll}
             & A_d-\tilde{D}^{0,\alpha_d}_{j+1}, \quad  j=0, \\
             & -\tilde{D}^{0,\alpha_d}_{j+1}, \qquad\quad  j>0.
        \end{array}\right.
\]
\hfill $\circ$
\end{theorem}

\begin{proof}
After applying the combined control law~\eqref{eq:combined_control} to achieve the desired network parameters $(\alpha_d, A_d)$, the network dynamics are governed by~\eqref{eq:defractionalized}. This can be written recursively as
\[
x[k+1] = \sum_{j=0}^{k} \tilde{A}^{0,\alpha_d}_{d,j} x[k-j] + B\tilde{u}[k].
\]

Following standard controllability analysis for this time-varying network, we can express the state evolution as
\[
x[k] = G^{0,\alpha_d}_{k}x_0 + \sum_{j=0}^{k-1} G^{0,\alpha_d}_{k-1-j}B\tilde{u}[j].
\]

Setting $k=T$ and requiring $x[T] = x_d$ yields the input sequence~\eqref{eq:state_control_solution}.
\end{proof}

\begin{remark}\label{rmk:feasibility_conditions}
The feasibility of simultaneously achieving the desired network parameters $(\alpha_d, A_d)$ and state target $x_d$ depends on:
\begin{enumerate}[label=(\roman*)]
    \item The individual conditions from Theorems~\ref{thm:theorem1} and~\ref{thm:theorem2} for parameter steering;
    \item If $\operatorname{rank}(C^{0,\alpha_d}_{T})=n$, then $x_d-G^{0,\alpha_d}_{T}x_0 \in \mathrm{Im}(C^{0,\alpha_d}_{T})$ and the equation admits solutions; one such solution of minimum Euclidean norm is given by~\eqref{eq:state_control_solution} with the left pseudoinverse ${C^{0,\alpha_d}_{T}}^\dag$ defined above. Conversely, if this rank condition fails, there exist target states $x_d$ that cannot be reached in $T$ steps.
\end{enumerate}
\hfill $\circ$
\end{remark}

\begin{remark}\label{rmk:lti_consistency}
When all the elements of $\alpha_d$ are equal to zero (i.e., the desired network is linear time-invariant), the fractional reachability matrix reduces to the classical controllability matrix, ensuring consistency with classical linear control theory
\[
C^{0,0}_{T} = \begin{bmatrix}
    B & A_d B & A_d^2 B & \ldots & A_d^{T-1} B
\end{bmatrix}.
\]
\hfill $\circ$
\end{remark}

\subsection{Physical Aware Energy-Constrained Control}\label{subsec:energy_qp}

While Theorem~\ref{thm:theorem4} provides the theoretical foundation for simultaneous parameter and state control, its direct implementation must account for physical and computational limitations. In particular, real-world systems are subject to actuator saturation, state safety bounds, and the need to truncate the infinite memory inherent in fractional-order dynamics.

In this subsection we formulate an optimization-based control design that explicitly incorporates these constraints. The resulting problem captures actuator limits, finite-memory approximations, and truncation-induced uncertainty, while remaining consistent with the reachability conditions and structural properties established in the previous sections.

Given the desired network configuration $\mathcal{N}_w(\alpha_d, A_d, B)$ achieved through the combined control law~\eqref{eq:combined_control}, we seek to drive the network state from initial condition $x_0 = x[0]$ to the desired terminal state $x_d = x[T]$ while minimizing energy consumption and respecting actuator constraints.

Following Proposition~\ref{prop:defractionalization}, the state evolution under the desired network parameters is governed by~\eqref{eq:defractionalized}. However, practical implementation necessitates truncating the infinite memory sum to a finite horizon $J \in \mathbb{N}$, introducing a bounded truncation error $w_J[k] \in \mathbb{R}^n$.

It follows that, given the maximum allowable actuator input magnitude $u_{\max} \in \mathbb{R}^+$, the memory truncation horizon $J$, $\alpha_i\ge \underline{\alpha}>0$ for all $i$, the maximum allowable state magnitude $B_x \in \mathbb R^+$, the final state value $x_d \in \mathbb{R}^n$ and the dynamic of the network, the energy-constrained terminal steering problem is formulated as
\begin{equation}
\begin{aligned}
\min_{\{\tilde{u}[k]\}_{k=0}^{T-1}}\ & \sum_{k=0}^{T-1}\|\tilde{u}[k]\|_2^2 \\[0.5em]
\end{aligned}
\label{eq:min_problem}
\end{equation}

% \begin{equation}
% \begin{aligned}
% \text{subject}& \ \text{to} \\[0.3em]
% & x[T] = x_d, \\[0.3em]
% & \|x[k]\|_{\infty} \leq B_x, \quad k = 0, \ldots, T-1, \\[0.3em]
% & \|\tilde{u}[k]\|_{\infty} \leq u_{\max}, \quad k = 0, \ldots, T-1, \\[0.3em]
% & x[k+1] = \sum_{j=0}^{\min\{k,J\}} \tilde{A}^{0,\alpha_d}_{d,j} x[k-j] + B\tilde{u}[k] + w_J[k], \\[0.3em]
% &\|w_J[k]\| \leq \varepsilon,  \quad k = J+1, \ldots, T-1,
% \end{aligned}
% \label{eq:energy_constrained}
% \end{equation}
subject to
$$x[T] = x_d,$$
$$\|x[k]\|_{\infty} \leq B_x, \quad k = 0, \ldots, T-1,$$
$$\|\tilde{u}[k]\|_{\infty} \leq u_{\max}, \quad k = 0, \ldots, T-1,$$
$$x[k+1] = \sum_{j=0}^{\min\{k,J\}} \tilde{A}^{0,\alpha_d}_{d,j} x[k-j] + B\tilde{u}[k] + w_J[k],$$
$$\|w_J[k]\| \leq \varepsilon, \quad k = J+1, \ldots, T-1.$$
\noindent where $\{\tilde{u}[k]\}_{k=0}^{T-1} \subset \mathbb{R}^p$ is the sequence of control inputs to be optimized. 

Here $w_J[k] \in \mathbb{R}^n$ denotes the truncation error, explicitly given by
\begin{align}
    w_J[k] &= -\sum_{j=J+1}^{k} \tilde{D}^{0,\alpha_d}_{j+1} x[k-j],
    \label{eq:truncation_error}
\end{align}
and we choose $\varepsilon>0$ in order to preserve the bound of the form $\varepsilon \le C J^{-\underline{\alpha}+1}\ln J$ (for some constant $C>0$) following from the decay estimate in Proposition~1.

\begin{remark}\label{rmk:optimization_formulation}
The problem formulated in~\eqref{eq:min_problem} is a quadratic program (QP) with linear constraints when the truncation error bound is treated as a robustness margin. The objective function minimizes the total control energy $\sum_{k=0}^{T-1}\|\tilde{u}[k]\|_2^2$, which is critical for practical applications where energy resources are limited.
\hfill $\circ$
\end{remark}

\begin{remark}\label{rmk:memory_truncation_choice}
The choice of truncation horizon $J$ involves a trade-off between computational complexity and control accuracy. A larger $J$ reduces the truncation error $\varepsilon$ but increases the computational burden of maintaining state history and evaluating the dynamics constraint. In practice, $J$ is selected based on the desired accuracy level and the decay rate determined by $\underline{\alpha}$. 

A rule of thumb is to select $\varepsilon$ such that it matches the inherent noise level in the system. Since real-world systems are invariably affected by process and measurement noise, there exists a fundamental limit below which distinguishing between truncation errors and noise-induced variations becomes impossible. Therefore, $\varepsilon$ should be chosen such that the memory truncation error is comparable to or smaller than the expected noise magnitude. Specifically, if the system noise is bounded by a factor $\sigma_v$, then selecting $\varepsilon \approx \sigma_v$ ensures that the truncation error remains within the noise floor, making further reduction of $J$ computationally inefficient without meaningful improvement in control accuracy.
\hfill $\circ$
\end{remark}

\begin{remark}\label{rmk:feasibility_robust}
The inclusion of the bounded truncation error $w_J[k]$ ensures that the problem in~\eqref{eq:min_problem} accounts for model uncertainty arising from finite memory implementation. This formulation naturally extends to robust optimal control frameworks where additional uncertainties in network parameters $(\alpha_d, A_d)$ may be present.
\hfill $\circ$
\end{remark}

%% file: 4_Numerical_Simulation.tex
%
%
%       Numerical Simulations
%
%
% To illustrate the utilization of the proposed results, we considered the same setting as in~\cite{Reed2023} where the first 6 channels of electrocorticography data from patient HUP64 ictal block 1 in the International Epilepsy Electrophysiology Portal~\cite{Wagenaar2013}. The data was recorded at a sampling rate of 512 Hz, pre-processed, and marked by clinical experts. 
\begin{figure*}
    \centering
    \includegraphics[width=0.7\textwidth]{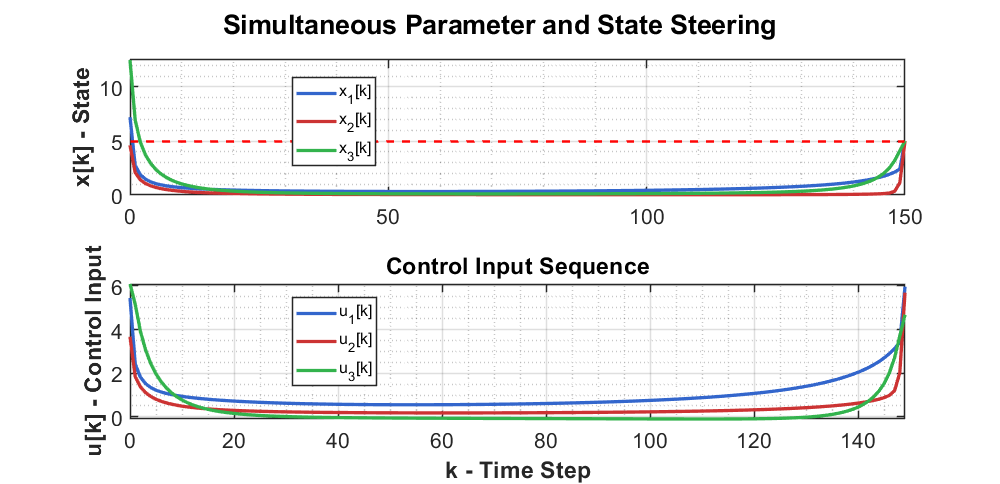}
    \vspace{-0.3cm}
    \caption{
    Results of simultaneous parameter and state steering for the 3-dimensional DTLFON. The top plots show the state dynamics for each component ($x_1$, $x_2$, $x_3$), demonstrating the successful transition from initial state $x_0^T = [7.2, 4.6, 12.5]$ to desired final state $x_d^T = [5.0, 5.0, 5.0]$. The bottom plots show the corresponding control input sequences $u_1[k]$, $u_2[k]$, and $u_3[k]$ applied over the 150-step control horizon to achieve both network parameter steering and state control.
    }
    \label{fig:param_state_steering}
\end{figure*}

Section IV demonstrates the practical effectiveness of our theoretical framework through three comprehensive examples. First, we present a detailed case study of simultaneous parameter and state steering in a 3-dimensional DTLFON, where we transition from an initial linear time-invariant configuration to a fractional-order network while simultaneously guiding the network state from a given initial condition to a desired target state over 150 time steps. 

Second, we validate the robustness and generalizability of our approach through extensive testing on artificially generated networks with three well-established topologies: Erdős-Rényi random networks, Barabási-Albert scale-free networks, and Watts-Strogatz small-world networks, scaling our analysis to 15×15 dimensional systems to demonstrate computational feasibility on realistic network sizes. 

Third, we showcase a clinically motivated application using real electrocorticography data from epilepsy patient HUP64, where we demonstrate how to steer a 6-channel brain network from pre-seizure dynamics (characterized by specific fractional exponents and coupling patterns) toward seizure-state configurations, illustrating the potential for therapeutic interventions through precise fractional parameter control.

\subsection{Example of Network Parameter Steering with Simultaneous State Control}
We consider the network transition from an initial LTI configuration (represented by choosing $\alpha$ as a zero vector) to a fractional-order configuration. The remaining parameter values and states are chosen randomly to create a toy example
\begin{align*}
\alpha &= \begin{bmatrix} 0 & 0 & 0 \end{bmatrix}, \\
A &= \begin{bmatrix}
-0.8 & 0.12 & 0.2 \\
-0.5 & 0.41 & 0.01 \\
0.6 & -0.02 & -0.27
\end{bmatrix}, \quad B = \begin{bmatrix}
1 & 0 & 0 \\
0 & 1 & 0 \\
0 & 0 & 0
\end{bmatrix},
\end{align*}
to the desired configuration
\begin{align*}
\alpha_d &= \begin{bmatrix} 0.5 & 0.1 & 0.8 \end{bmatrix}, \\
A_d &= \begin{bmatrix}
-0.1 & 0.1 & -0.05 \\
-0.1 & 0.1 & 0.15 \\
0.05 & -0.2 & -0.2
\end{bmatrix}, \quad B = \begin{bmatrix}
1 & 0 & 0 \\
0 & 1 & 0 \\
0 & 0 & 0
\end{bmatrix}.
\end{align*}
Simultaneously, we aim to steer the network state from initial state $x_0$ to the desired final state $x_d$, where:
\begin{align*}
x_0^T &= \begin{bmatrix} 7.2 & 4.6 & 12.5 \end{bmatrix}, \\
x_d^T &= \begin{bmatrix} 5.0 & 5.0 & 5.0 \end{bmatrix}.
\end{align*}
We consider a case in which a sequence of $T = 150$ input values is required to achieve both the network parameter transition $\mathcal{N}_w(\alpha,A,B) \to \mathcal{N}_w(\alpha_d,A_d,B)$ and the state transition $x_0 \to x_d$. 
The control approach proceeds in two stages: we first apply the results from Theorem~\ref{thm:theorem4} to compute the optimal input sequence that guides the DTLFON from the initial state $x_0$ to the desired final state $x_d$ for the target network configuration $\mathcal{N}_w(\alpha_d,A_d,B)$. This involves computing the fractional controllability matrix ${C^{0,\alpha_d}_{T}}$ and determining the minimum-energy input sequence.

In a second stage, we used Theorem~\ref{thm:theorem3} to determine the control input needed to transition the DTLFON from the network configuration $\mathcal{N}_w(\alpha,A,B)$ to the desired configuration $\mathcal{N}_w(\alpha_d,A_d,B)$ while preserving the previous input which allows to simultaneously achieving the state transition.
The resulting control input sequence $\{u[k]\}_{k=0}^{T}$ successfully achieves both objectives, steering the network parameters from $(\alpha,A)$ to $(\alpha_d,A_d)$ and the network state from $x_0$ to $x_d$ (see Fig~\ref{fig:param_state_steering}).

\subsection{Artificial Networks Analysis}

To provide generality and validate the robustness of our proposed framework across different network topologies, we conducted extensive testing on three well-established network architectures: Erdős-Rényi, Barabási-Albert, and \mbox{Watts-Strogatz} networks. These network models were selected due to their prominence in the literature and their frequent adoption in network analysis studies. 

Particularly relevant to our neuroscience applications, the Watts-Strogatz model represents small-world networks, which characterize the organizational topology of brain anatomical and functional networks. Such networks exhibiting both dense local clustering and short path lengths that support the segregated and integrated information processing observed in neural systems~\cite{Bassett2006}. 

\begin{figure}[htbp]
    \centering
    \includegraphics[width=0.48\textwidth]{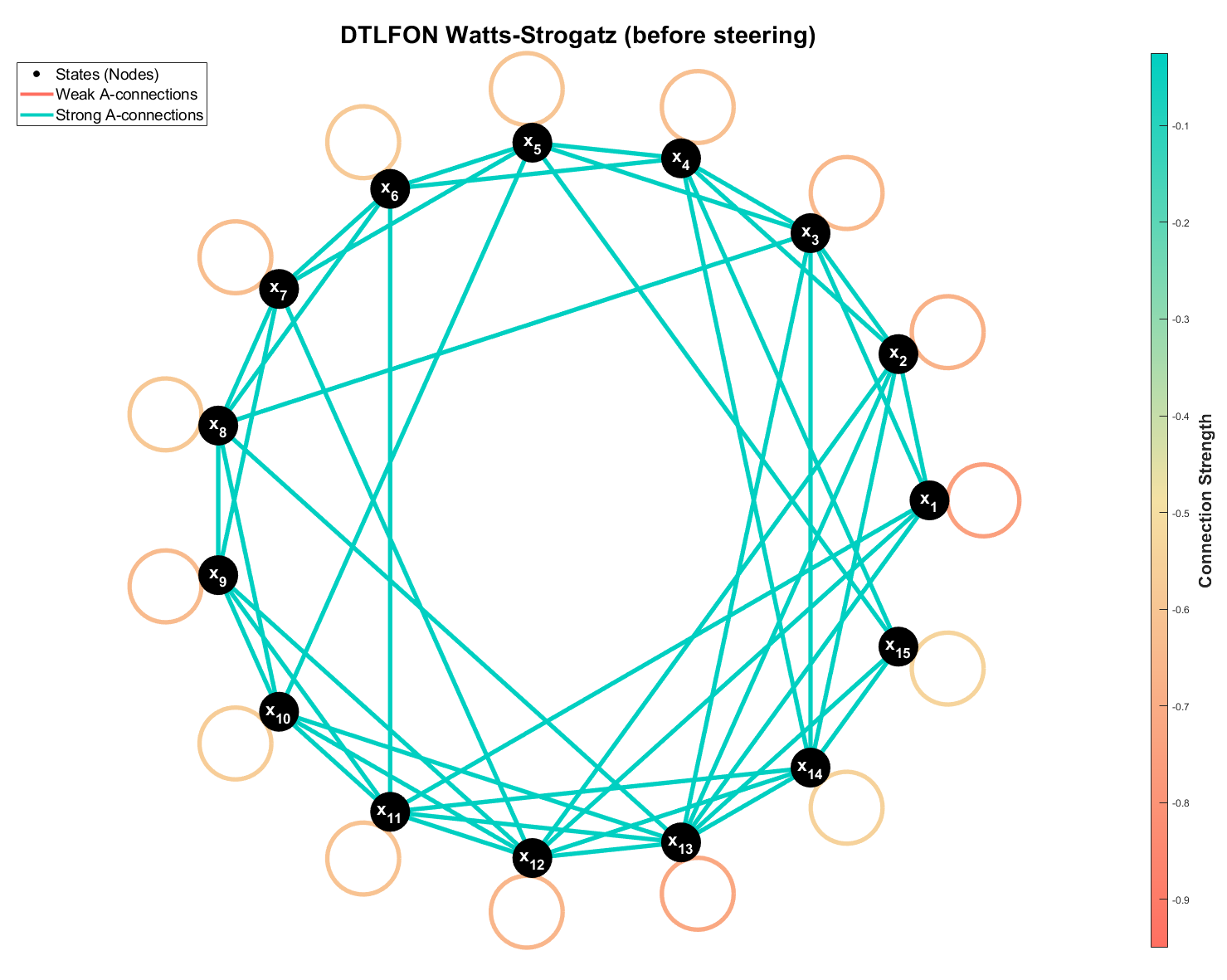}
    \caption{Visualization of a Watts-Strogatz network instance before parameter steering. Nodes represent the individual network states, while edges between nodes represent the coupling strength defined by the matrix $A$. The colorbar indicates connection strength: light blue denotes weak connections, sand color represents intermediate connections, and coral color indicates strong connections.}
    \label{fig:network_topology}
\end{figure}

An example of this network topology is visualized in Figure~\ref{fig:network_topology}, where the coupling structure between network states is represented through weighted connections. For each network topology, we constructed 15×15 dimensional systems, representing networks five times larger than the illustrative example presented above. This scaling allowed us to assess the computational feasibility and effectiveness of our control strategies on more complex, realistic network configurations while maintaining tractable analysis.

Notably, in this analysis we employed a more general input matrix $B$ that is not the identity matrix, as illustrated in Figure~\ref{fig:watts_strogatz_analysis} (Left). This demonstrates that our framework does not require full actuation on all network nodes; rather, it is sufficient that the input matrix $B$ satisfies the conditions established in Theorem~\ref{thm:theorem3} for combined parameter steering.

\begin{figure*}[htbp]
    \centering
    
    % Watts-Strogatz network analysis with three subplots
    \includegraphics[width=1\textwidth]{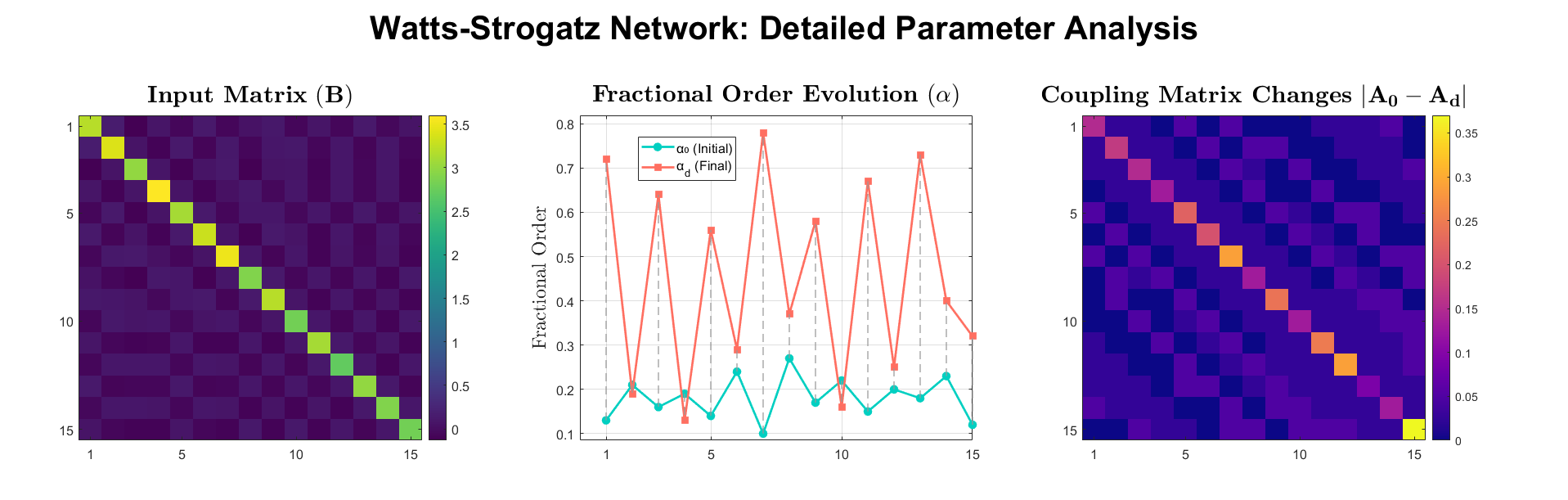}
    
    \caption{Watts-Strogatz network parameter analysis showing the steering process across three complementary visualizations. 
    \textbf{Left:} Input matrix (B) displaying the intensity with which control inputs are injected into each network node. Notably, this matrix remain the same throughout the entire steering process.
    \textbf{Center:} Fractional order evolution illustrating the progression from initial values $\alpha_0$ to final values $\alpha_d$, where most exponents were increased, indicating that the network moved to a state of stronger memory conditions during the steering process. 
    \textbf{Right:} Absolute difference heatmap showing $|A_0 - A_d|$ for each matrix entry, where $A_0$ represents the initial value of the spatial-temporal coupling matrix $A$ and $A_d$ the target value. This plot shows that the main contributions to coupling matrix changes occurred along the main diagonal elements, reflecting the restructuring of self-coupling dynamics during parameter steering.}
    
    \label{fig:watts_strogatz_analysis}
\end{figure*}

Across all tested network architectures, the fractional-order networks exhibited behavior consistent with our theoretical predictions. When subjected to input sequences designed according to the respective theorems, each network successfully demonstrated the ability to steer both network parameters and network states from predetermined initial conditions to arbitrary target configurations. An example of this parameter steering process can be seen in Figure~\ref{fig:watts_strogatz_analysis} (Center, Right), which illustrates the main parameters of the Watts-Strogatz network before and after steering.

The consistency of these results across diverse topological structures reinforces the general applicability of our proposed control framework and validates its potential for implementation in various real-world scenarios, particularly those involving complex brain network dynamics where topology plays a crucial role in determining overall network behavior.

\subsection{Example of Network Parameter Steering in the Brain}

To demonstrate the clinical relevance of our theoretical framework, we apply our methodology to brain network parameters derived from real electrocorticography (ECoG) data from an epileptic patient. This example illustrates how fractional-order network parameters can characterize different brain states and how our control framework can potentially guide therapeutic interventions.

We utilize network parameters extracted from electrocorticography data of patient HUP64 (ictal block 1) from the International Epilepsy Electrophysiology Portal, as analyzed and reported in~\cite{Reed2023}. The original data were recorded at a sampling rate of 512 Hz using clinical-grade electrodes implanted directly on the brain surface for surgical evaluation. Clinical experts marked seizure onset and termination times, and the recordings underwent standard pre-processing procedures.
Reed et al. performed parameter estimation for discrete-time linear fractional-order network (DTLFON) models using the first 6 channels of the recording, which correspond to electrodes positioned over cortical regions showing clear seizure activity.

We adopt the network parameters estimated at two critical time points: 12 seconds before seizure onset (representing stable pre-ictal dynamics), with network $\mathcal{N}_{w}(\alpha,A,B)$, and 48 seconds after seizure initiation (representing established seizure activity), with network $\mathcal{N}_w(\alpha_d,A_d,B)$. These represent two distinct brain configurations from the epilepsy dataset.

The pre-ictal network parameters are the following:
\[
{\scriptsize A = \begin{bmatrix} 
-0.0479 &0.1919&0.1358&-0.3619&0.0249&-0.0736\\
-0.0756&0.2065&0.1274&-0.4029&0.0150&-0.0175\\
   -0.1016&0.2417&0.0904  & -0.6914&0.0453&0.1177\\
   -0.1157&0.1818&0.0733&0.2628&0.0862 & -0.0117\\
0.0047&0.1554&0.0656  & -0.4384&0.2923  & -0.0307\\
0.0756&0.3646&0.0106  & -0.7247&0.0531 &  -0.0169
\end{bmatrix},}
\]
\[
B=\mathbb{I}_{6\times 6}, \text{ and }
\]
{\footnotesize \[ 
\alpha = \begin{bmatrix}0.6228&0.6189&0.4996&0.5375&0.9873&0.6843\end{bmatrix}.
\] }

Moreover, the target network parameters are:
{\scriptsize \[ A_d = \]}
{\scriptsize \[
\begin{bmatrix}
-0.0261 &0.0192&-0.0679&-0.0977&-0.2879&-0.0260\\
-0.1454&0.0623&-0.0168&0.1685&-0.1814&-0.3644\\
   -0.1690&-0.0720&0.0637  & -0.2994&-0.0289&0.2221\\
   -0.1157&-0.0435&0.0039&0.6502&-0.2592 & -0.0048\\
-0.1116&-0.1056&0.0591  & -0.0344&0.1489  & 0.1556\\
-0.2673&-0.0428&-0.1412  & 0.3054&-0.2701 &  -0.1832
\end{bmatrix},
\]}
{\footnotesize \[
 \alpha_d = \begin{bmatrix}
    0.0945&0.0478&0.2429&0.1801&0.1835&0.2648
\end{bmatrix}.
\] }

First, we apply Theorem~\ref{thm:theorem1} to steer the coupling matrix from $A$ to $A_d$. The required feedback matrix $K$ is:
{\scriptsize \[
K = A - A_d = 
\]}
{\scriptsize \[
\begin{bmatrix}
-0.0218 & 0.1727 & 0.2037 & -0.2642 & 0.3128 & -0.0476 \\
0.0698 & 0.1442 & 0.1442 & -0.5714 & 0.1964 & 0.3469 \\
0.0674 & 0.3137 & 0.0267 & -0.3920 & 0.0742 & -0.1044 \\
0.1025 & 0.2253 & 0.0694 & -0.3874 & 0.3454 & -0.0069 \\
0.1163 & 0.2610 & 0.0065 & -0.4040 & 0.1434 & -0.1863 \\
0.3429 & 0.4074 & 0.1518 & -1.0301 & 0.3232 & 0.1663
\end{bmatrix}.
\]}

Having established the coupling matrix steering mechanism, we now focus on the DTLFON fractional exponent control. 
The choice of fractional exponents $\alpha_d$ and coupling matrix $A_d$ significantly influences the network's dynamic behavior, and certain parameter combinations can lead to undesired network responses that may require correction through appropriate control interventions.

In order to achieve more favorable network dynamics, we steer $\alpha$ to a new value, denoted as $\alpha_d$, by adopting the result in Theorem~\ref{thm:theorem2}. In Fig.~\ref{fig:test1}, we present four distinct scenarios that showcase different behavioral regimes when the network is subjected to various input sequences.

In the first plot at the top, we observe the dynamics of a DTLFON characterized by fractional exponent $\alpha$ and coupling matrix $A_d$ when exposed to a unitary step input. This configuration exhibits pronounced oscillatory behavior with growing amplitude, representing an undesired operational regime that motivates the need for parameter control.

Progressing to the second plot, the focus shifts to a DTLFON with modified fractional exponent $\alpha_d$, again responding to a unitary step input. This comparison demonstrates how altering the fractional exponent fundamentally transforms the network's response characteristics, resulting in more controlled and bounded dynamics compared to the first scenario.

The third plot further illustrates the effectiveness of our approach by depicting the response of the original DTLFON (with fractional exponent $\alpha$) to a carefully designed input sequence defined in accordance with Theorem~\ref{thm:theorem2}. This scenario reveals how strategic input design can compensate for unfavorable parameter configurations, demonstrating the power of our theoretical framework.

Finally, the fourth plot offers a comprehensive view by combining elements of the preceding scenarios. It features the response of the original DTLFON to an input sequence that begins with a unitary step for the initial time period (from 0 to 15), subsequently transitioning to the theoretically-designed sequence from Theorem~\ref{thm:theorem2} for the remainder of the observation period (from 16 to 100). This demonstrates the transition from uncontrolled to controlled behavior through strategic input modification.

It is important to note that at this stage we have not yet employed the results from Theorem~\ref{thm:theorem4} to simultaneously steer the network state. As a consequence, in the final plot the state reaches an unpredicted value that differs from the previous scenarios, highlighting the need for combined parameter and state control strategies.

\begin{figure}
    \centering
    \includegraphics[scale=0.27]{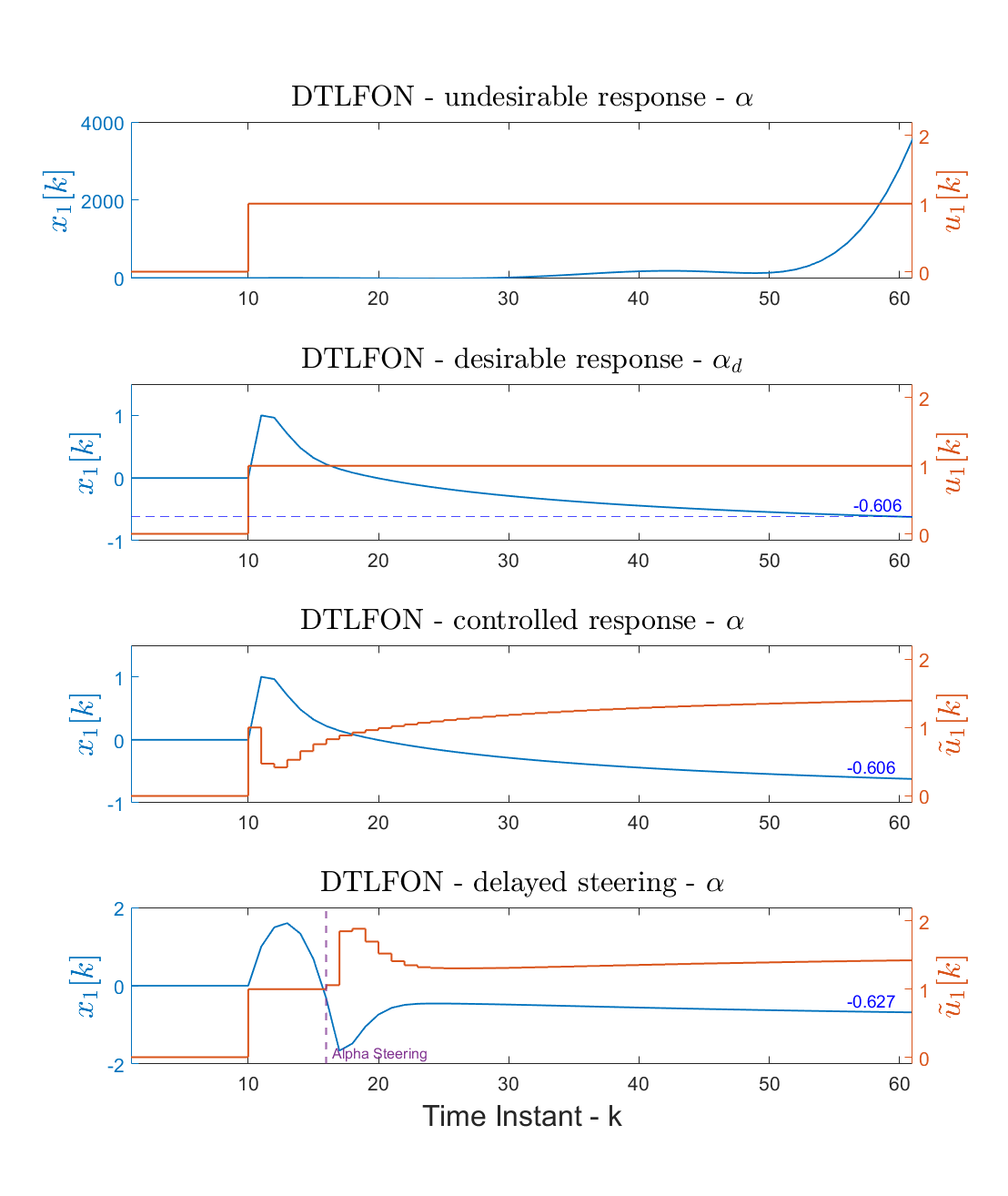}
    \vspace{-0.5cm}
    \caption{
    Response of DTLFON with different parameters and input sequences. To highlight the network response, the first state dynamic is depicted. From top: (1) Undesirable response of network $\mathcal{N}_w(\alpha, A, B)$ to a unit step input due to parameter configuration; (2) Desirable response of network $\mathcal{N}_w(\alpha_d, A_d, B)$ to a unit step input, demonstrating favorable dynamics from appropriate parameter selection; (3) Desirable response achieved by steering the network from $\mathcal{N}_w(\alpha, A, B)$ to $\mathcal{N}_w(\alpha_d, A_d, B)$ through an input sequence constructed according to Theorem~\ref{thm:theorem3}; (4) Response when input injection is delayed, beginning with unit step input ($t \in[0,15]$) then transitioning to the Theorem~\ref{thm:theorem3} sequence ($t > 15$, red dotted line). The delayed injection illustrates how altering $\alpha$ and $A$ affects the network state, highlighting the need for Theorem~\ref{thm:theorem4} to simultaneously control both parameters and state.
    }
    \label{fig:test1}
\end{figure}

%% file: 5_Conclusion.tex
%
%
%       Conclusion
%
%
In our work we shown that fractional-order networks can be systematically reconfigured by controlling both fractional exponents and coupling matrices through strategically computed input sequences, while providing the necessary conditions and computational methodologies to determine these sequences. 
The practical applicability of our theoretical framework was demonstrated through numerical simulations, particularly in neuroscience applications where abnormal brain network dynamics, such as those observed during epileptic seizures, can potentially be guided toward healthier configurations through precise parameter control.
Our methodology allows for simultaneous parameter and state steering as in Theorem~\ref{thm:theorem2} and Theorem~\ref{thm:theorem3}, while simultaneously achieving state control objectives as in Theorem~\ref{thm:theorem4}. 
An interesting extension of this work would be to explore scenarios where constraints are present, which would necessitate formulating a non-linear optimization problem to compute the optimal sequences of coupling matrix $A$, fractional exponents $\alpha$, and network state $x$ that collectively achieve the desired configuration within the available energy budget.
Future research directions include developing constrained control strategies, integrating machine learning techniques for real-time adaptation, investigating robust control methodologies to handle uncertainties, and extending the framework to other neurological conditions beyond epilepsy management.